\documentclass[journal]{IEEEtran}
 
 \usepackage[T1]{fontenc}
 
 \usepackage{amsmath}
 \usepackage{amsthm}
 \usepackage{bm}
 \usepackage{mathtools}
 \usepackage{nccmath}
 \usepackage[cmintegrals]{newtxmath}
 
 \usepackage{soul}
 \usepackage{xcolor}
 \usepackage{xspace}
 
 \theoremstyle{plain}
 \newtheorem{theorem}{Theorem}

 \theoremstyle{definition}

 \usepackage{graphicx}
 \usepackage{stfloats}
 \usepackage{float}
 
 \usepackage{algorithm}
 \usepackage{algpseudocode}
 
 \usepackage{tabularx}
 \usepackage[inline]{enumitem}
 
 \usepackage[sort,compress]{cite}
 
 \usepackage{dsfont}
 \usepackage{cases}
 \usepackage{setspace}
 \usepackage[framemethod=tikz]{mdframed}
 
 \usepackage[user]{zref}
 \usepackage[hidelinks]{hyperref}
 \usepackage{balance}
 
 \newcounter{revc}
 
 \makeatletter
 
 \zref@newprop{revsec}{}
 \zref@newprop{revcontent}{}
 
 \zref@addprop{main}{revsec}
 \zref@addprop{main}{revcontent}
 
 \newcommand{\revinu}[2]{%
 	\zref@setcurrent{revsec}{\thesection}%
 	\zref@setcurrent{revcontent}{#2}%
 	\refstepcounter{revc}%
 	\zlabel{#1}%
 	\label{#1}%
 	#2%
 }
 
 \newcommand{\revr}[2]{%
 	\zref@setcurrent{revsec}{\thesection}%
 	\zref@setcurrent{revcontent}{#2}%
 	\refstepcounter{revc}%
 	\zlabel{#1}%
 	\label{#1}%
 	\st{#2}%
 }
 
 \usepackage[caption=false,font=footnotesize]{subfig}
 \makeatother
 
 \definecolor{mycolor}{rgb}{0.122,0.435,0.698}
 
 \newmdenv[
 innerlinewidth=0.5pt,
 roundcorner=4pt,
 linecolor=mycolor,
 innerleftmargin=6pt,
 innerrightmargin=6pt,
 innertopmargin=6pt,
 innerbottommargin=6pt
 ]{mybox}
 
 \newcommand{\bmP}{\boldsymbol{\Psi}}

 \newcommand{\bmV}{\mathbf{V}}

 \def\BibTeX{%
 	{\rm B\kern-.05em%
 		{\sc i\kern-.025em b}\kern-.08em%
 		T\kern-.1667em%
 		\lower.7ex\hbox{E}\kern-.125emX}%
 }
\begin{document}

\title{Stacked Intelligent Metasurfaces Assisted UAV Communications}  
\author{Chandan~Kumar~Sheemar,~\IEEEmembership{Member,~IEEE}, Giovanni Iacovelli,~\IEEEmembership{Member,~IEEE}, Sourabh Solanki,~\IEEEmembership{Member,~IEEE},  \\ Wali Ullah Khan,~\IEEEmembership{Member,~IEEE}, Symeon Chatzinotas,~\IEEEmembership{Fellow,~IEEE}   
 \thanks{C. K. Sheemar, G. Iacovelli, W. U. Khan, and S. Chatzinotas are with the SnT, University of Luxembourg, (emails:\{chandankumar.sheemar,\hspace{0pt}giovanni.iacovelli,\hspace{0pt}waliullah.khan,\hspace{0pt}symeon.chatzinotas\}@uni.lu). Sourabh Solanki is with National Institute of Technology Warangal, TS, 506004, India (e-mail: \{ssolanki\}@nitw.ac.in).}
} 
 \maketitle
 
\begin{abstract}
In this paper, we investigate an unmanned aerial vehicle (UAV) communication system assisted by stacked intelligent metasurfaces (SIMs), which enable programmable wave-domain signal processing through multiple cascaded metasurface layers. By shifting part of the beamforming functionality from the RF/digital domain to the electromagnetic domain, SIMs allow the realization of energy-efficient hybrid beamforming architectures suitable for aerial platforms. We formulate the joint design of digital precoding, SIM phase configuration, and UAV positioning for multi-user downlink sum-rate maximization. To solve the resulting non-convex problem, we develop an alternating optimization framework that guarantees monotonic improvement of the objective. Numerical results demonstrate that the proposed SIM-assisted architecture significantly improves spectral efficiency while maintaining low hardware complexity, and highlight the impact of the number of metasurface layers and size of each layer on system performance.
\end{abstract}

\IEEEpeerreviewmaketitle

\section{Introduction} \label{Intro}
 
 \IEEEPARstart{U}{nmanned aerial vehicles} (UAVs) are a promising enabler for next-generation wireless networks due to their flexible deployment, strong line-of-sight (LoS) connectivity, and rapid coverage capabilities \cite{sheemar2026joint}. However, supporting multi-user communications with large antenna arrays is challenging, as fully digital beamforming requires many radio frequency (RF) chains and high power consumption \cite{sheemar2025near}, which conflicts with UAV energy and payload constraints. 

Recently, stacked intelligent metasurfaces (SIMs) have been proposed as an energy-efficient architecture for realizing wave-domain signal processing through multiple programmable metasurface layers \cite{sheemar2026survey}. By manipulating the phase of incident electromagnetic waves across cascaded layers, SIMs can perform analog beamforming and spatial filtering directly in the propagation domain. This enables complex wavefront shaping with significantly fewer RF chains, thereby reducing hardware complexity and power consumption \cite{an2023stacked,an2024stacked_WC}. Compared to conventional single-layer designs, the layered structure provides enhanced degrees of freedom for controlling wave propagation \cite{papazafeiropoulos2025ergodic_outage_3}, making SIM-assisted UAV systems a promising architecture for next-generation aerial networks.

Despite its potential, SIM-assisted UAV communication has received very limited attention \cite{fan2025joint,zarini2025orchestration}. In \cite{fan2025joint}, an uplink SIM-assisted UAV system is studied, jointly optimizing user--UAV association, UAV locations, and SIM phase shifts; however, the architecture relies only on analog SIM processing, and its formulation does not extend to the hybrid downlink setting, where digital precoding and the cascaded SIM response are coupled through the effective multi-user channel. In \cite{zarini2025orchestration}, the SIM response, UAV position, and transmit power are optimized via meta-learning-based deep reinforcement learning; digital beamforming is again omitted, and the data-driven solution provides no analytical treatment of the hybrid multi-layer design. UAV-mounted SIMs have also been exploited for wireless sensing in \cite{lin2025uav}, where the SIM acts as a diffractive neural network for low-complexity direction-of-arrival estimation; while confirming the feasibility of airborne wave-domain processing, this work does not address multi-user downlink beamforming or UAV placement. Consequently, analytically tractable joint frameworks for SIM-assisted hybrid UAV system remain unexplored.

To address this gap, we develop a unified framework for SIM-assisted multi-user downlink UAV communications, where the SIM performs analog wave-domain processing while low-dimensional digital beamforming is carried out at baseband. We formulate a sum-rate maximization problem that jointly optimizes the digital beamformer, the SIM phase shifts, and the 3D UAV position, and solve it via a structure-exploiting alternating optimization algorithm. Specifically, the digital beamformer is obtained in closed form via zero-forcing (ZF) precoding, the SIM coefficients are optimized using block coordinate descent (BCD) with closed-form element-wise updates based on a quadratic surrogate, and the UAV position is refined through projected gradient ascent (PGA) using analytically derived rate gradients. 
The proposed algorithm guarantees monotonic improvement of the objective, and its limit points satisfy blockwise first-order optimality conditions.
Numerical results demonstrate significant performance gains of SIM-assisted UAV systems compared to conventional single-layer designs under various system configurations, including transmit power, number of SIM layers, and metasurface size per layer\footnote{\emph{Notations:} Scalars, vectors, and matrices are denoted by italic, bold lowercase, and bold uppercase letters, respectively. $(\cdot)^T$, $(\cdot)^H$, and $(\cdot)^*$ denote transpose, Hermitian transpose, and complex conjugate, respectively. $\|\cdot\|$ denotes the Euclidean norm and $\|\cdot\|_F$ the Frobenius norm. $\mathrm{diag}(\cdot)$ denotes a diagonal matrix formed from a vector, while $\Re\{\cdot\}$ denotes the real part of a complex number. $[\mathbf{x}]_m$ denotes the $m$-th element of vector $\mathbf{x}$}.
 \begin{figure}
    \centering
    \includegraphics[width=0.5\linewidth,height=4cm]{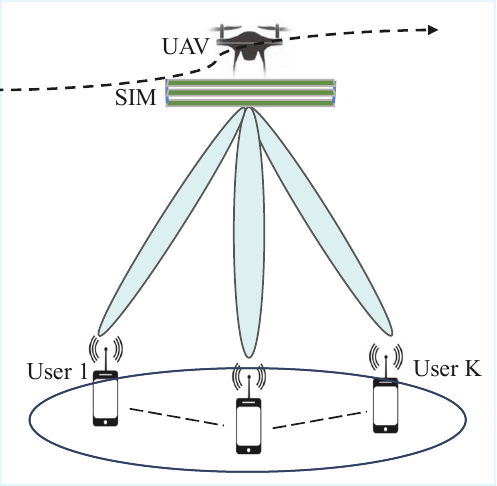}
    \caption{SIM-assisted UAV communications.}
    \label{fig:placeholder}
\end{figure}
 
\section{System Model and Problem Formulation}
\subsection{System Model}
We consider the downlink of a SIM mounted UAV multi-user communication system serving $K$ single-antenna ground users. The transmitter employs hybrid beamforming consisting of digital precoding and a SIM implementing analog wave-domain beamforming. The number of RF chains is $N=K$, enabling a low hardware complexity design. The SIM comprises $L$ metasurface layers indexed by $\mathcal{L}=\{1,\ldots,L\}$, fed by the $N$ RF chains. Each layer contains $M$ meta-atoms arranged in a uniform planar array satisfying $M\ge N$, with index set $\mathcal{M}=\{1,\ldots,M\}$. The transmission coefficient imposed by the $m$-th meta-atom on the $l$-th metasurface layer is
$
\phi_{l,m}=e^{j\theta_{l,m}}, \qquad \theta_{l,m}\in[0,2\pi),
 $
where $\theta_{l,m}$ denotes the adjustable phase shift. Let
 $
\boldsymbol{\phi}_l=
[\phi_{l,1},\phi_{l,2},\ldots,\phi_{l,M}]^T
 $
and
 $
\mathbf{\Phi}_l=\text{diag}(\boldsymbol{\phi}_l)\in\mathbb{C}^{M\times M}
 $
represent the transmission coefficient vector and matrix of the $l$-th SIM layer.

All metasurfaces follow an isomorphic lattice structure and are modelled as uniform planar arrays \cite{an2023stacked}. The distance between the $m$-th and $\tilde m$-th meta-atoms on the same transmit metasurface is $
r_{m,\tilde m}= \Delta \sqrt{(m_z-\tilde m_z)^2+(m_x-\tilde m_x)^2}
$ where $\Delta$ denotes the inter-element spacing. The indices of the $m$-th meta-atom are defined as $ 
m_z=\left\lceil\frac{m}{m_{\max}}\right\rceil$, with $
m_x=\text{mod}(m-1,m_{\max})+1 $
with $M=m_{\max}^2$. All metasurface layers are assumed parallel with uniform spacing. Let $D_t$ denote the thickness of the SIM and $d_t=D_t/(L-1)$ the spacing between adjacent metasurfaces. The transmission distance from the $\tilde m$-th meta-atom on layer $l-1$ to the $m$-th meta-atom on layer $l$ is $ 
r_{m,\tilde m}^{(l)}=\sqrt{r_{m,\tilde m}^2+d_t^2}, $ with $l\in\mathcal{L}\setminus\{1\}$.  Based on Rayleigh-Sommerfeld diffraction theory, the propagation coefficient between these two meta-atoms is \cite{an2023stacked}:
\begin{equation}
w_{m,\tilde m}^{l}=
\frac{A_t\cos\chi_{m,\tilde m}^{l}}{r_{m,\tilde m}^{(l)}}
\left(
\frac{1}{2\pi r_{m,\tilde m}^{(l)}}
-\frac{j}{\lambda}
\right)
e^{j\frac{2\pi}{\lambda}r_{m,\tilde m}^{(l)}}
\end{equation}
where $A_t$ denotes the area of each meta-atom, $\lambda$ the wavelength, and $\chi_{m,\tilde m}^{l}$ the angle between the propagation direction and the metasurface normal. 
Let $\mathbf{W}_l\in\mathbb{C}^{M\times M}$, containing the entries $w_{m,\tilde m}^{l}$, denote the transmission coefficient matrix between the $(l-1)$-th and $l$-th metasurface layers, for $l\in\mathcal{L}\setminus\{1\}$. Analogously, $\mathbf{W}_1\in\mathbb{C}^{M\times N}$ denotes the propagation matrix from the $N$ RF feeds to the first layer, whose entries follow the same Rayleigh--Sommerfeld model with the corresponding feed--atom distances. Hence, since $\mathbf{\Phi}_1\mathbf{W}_1\in\mathbb{C}^{M\times N}$ and $\mathbf{\Phi}_l\mathbf{W}_l\in\mathbb{C}^{M\times M}$ for $l>1$, the equivalent SIM analog precoder is
$\boldsymbol{\Psi}=\mathbf{\Phi}_L\mathbf{W}_L\cdots\mathbf{\Phi}_1\mathbf{W}_1\in\mathbb{C}^{M\times N}$.
Let $\mathbf{s}\in\mathbb{C}^{K\times1}$ denote the transmitted symbol vector satisfying $\mathbb{E}[\mathbf{s}\mathbf{s}^H]=\mathbf{I}$. Digital beamforming is performed through $
\mathbf{V}\in\mathbb{C}^{N\times K}$. 
Thus, the transmitted signal after SIM processing is $
\mathbf{x} =\boldsymbol{\Psi}\mathbf{V}\mathbf{s}$.
Let $\mathbf{h}_k\in\mathbb{C}^{M\times1}$ denote the channel between the SIM output aperture and user $k$. Let $\mathbf{q} =[x_u,y_u,z_u]^T$ denote the 3D coordinates of the UAV and $\mathbf{u}_k =[x_k,y_k,0]^T$   denote the position of user $k$. The channel between the SIM output aperture and user $k$ is modeled by considering the propagation from each meta-atom of the last SIM layer to the user. Let $\mathbf{p}_m=[x_m,y_m,z_m]^T$ denote the position of the $m$-th meta-atom on the last SIM layer. The propagation distance between this meta-atom and user $k$ is given by:
\begin{equation}
d_{k,m}(\mathbf{q}) = \| \mathbf{q}+\mathbf{p}_m-\mathbf{u}_k \|_2 .
\end{equation}
Assuming line-of-sight (LoS) propagation, the channel coefficient between the $m$-th meta-atom and user $k$ is modeled as:
\begin{equation} \label{canale}
h_{k,m}(\mathbf{q})=
\sqrt{\beta_0}\, d_{k,m}(\mathbf{q})^{-\frac{\alpha}{2}}
e^{-j\frac{2\pi}{\lambda} d_{k,m}(\mathbf{q})},
\end{equation}
where $\beta_0$ is the channel gain at a reference distance of $1$ m, $\alpha$ denotes the path-loss exponent, and $\lambda$ is the carrier wavelength. Thus, the channel vector between the SIM output layer and user $k$ can be written as:
$
\mathbf{h}_k(\mathbf{q})=
\left[
h_{k,1}(\mathbf{q}),\,
h_{k,2}(\mathbf{q}),\,
\ldots,\,
h_{k,M}(\mathbf{q})
\right]^T
\in\mathbb{C}^{M\times1}
 $.
The received signal at user $k$ is therefore:
\begin{equation} \label{received_signal}
y_k=\mathbf{h}_k^H(\mathbf{q})\boldsymbol{\Psi}\mathbf{V}\mathbf{s}+n_k
\end{equation}
where $n_k\sim\mathcal{CN}(0,\sigma_k^2)$ is the additive white Gaussian noise. The resulting SINR of user $k$ is given by
\begin{equation}
\gamma_k(\bmP,\bmV,\mathbf{q})=
\frac{|\mathbf{h}_k^H(\mathbf{q})\boldsymbol{\Psi}\mathbf{v}_k|^2}
{\sum_{i\neq k}|\mathbf{h}_k^H(\mathbf{q})\boldsymbol{\Psi}\mathbf{v}_i|^2+\sigma_k^2}
\end{equation}
where $\mathbf{v}_k$ denotes the $k$-th column of $\mathbf{V}$.

\subsection{Problem Formulation}
Based on the above system model, the objective is to jointly optimize the digital beamformer, the SIM transmission coefficients, and the UAV position to maximize the system sum rate. The problem can be formulated as
\begin{subequations}\label{opt_problem}
\begin{align}
\max_{\substack{\mathbf{V},\,\{\mathbf{\Phi}_l\}, \mathbf{q}}} \quad & 
\sum_{k=1}^{K}
\log_2\!\Big(
1+ \gamma_k(\{\mathbf{\Phi}_l\},\bmV,\mathbf{q})
\Big)
\label{opt_problem_obj}
\\
\text{s.t.}\quad
& 
\text{Tr}\Big( \boldsymbol{\Psi} \mathbf{V} \mathbf{V}^H \boldsymbol{\Psi}^H  \Big)
\le P_{\max},
\label{opt_problem_power}
\\
& \theta_{\ell,m}\in[0,2\pi), \quad  \quad \forall \ell,m,
\label{opt_problem_phase}
\\
& j_{\min}\le j_u\le j_{\max}, \quad \text{where}\; j \in \{x,y,z\}.
\label{opt_problem_x}
\end{align}
\end{subequations}
In \eqref{opt_problem}, constraint \eqref{opt_problem_power} ensures that the total transmit power after hybrid beamforming does not exceed the maximum available transmit power $P_{\max}$. Constraints \eqref{opt_problem_phase}  represent the phase of the SIM elements, and constraint \eqref{opt_problem_x} restrict the UAV 3D deployment region by bounding its horizontal coordinates and altitude.

\section{Proposed Optimization Framework}

Problem~\eqref{opt_problem} is highly non-convex due to the coupling among the digital beamformer, the cascaded SIM structure, the nonlinear dependence on the UAV position, and the unit-modulus constraints. We address it via alternating optimization with closed-form ZF and BCD-based SIM updates, followed by PGA refinement of the UAV position.

\subsection{Digital Beamformer Design}
For fixed SIM coefficients $\{\mathbf{\Phi}_\ell\}$ and UAV position $\mathbf{q}$, define the effective channel for user $k$ as $\bar{\mathbf{h}}_k^H(\mathbf{q})=\mathbf{h}_k^H(\mathbf{q})\boldsymbol{\Psi}$ and stack the effective user channels into
$ \bar{\mathbf{H}}(\mathbf{q}) = [\bar{\mathbf{h}}_1^H(\mathbf{q}),\cdots,\bar{\mathbf{h}}_K^H(\mathbf{q})]^T$.
To suppress inter-user interference, we employ a ZF digital beamformer based on the Moore–Penrose pseudoinverse
\begin{equation}\label{digi_solu}
\mathbf{V}
=
\eta\,\bar{\mathbf{H}}^{\dagger}(\mathbf{q})
\end{equation}
where $(\cdot)^{\dagger}$ denotes the pseudoinverse. The normalization factor $\eta$ guarantees the transmit power constraint and is given by $
\eta
=
\sqrt{
\frac{P_{\max}}
{
\left\|\boldsymbol{\Psi}
\bar{\mathbf{H}}^{\dagger}(\mathbf{q}) 
\right\|_F^2
}
}$. 
This closed-form solution nulls the inter-user interference for given $\boldsymbol{\Psi}$ and $\mathbf{q}$; hence, it must be recomputed whenever the UAV moves or the SIM response is updated.

 \subsection{SIM Optimization}

For a fixed UAV position $\mathbf{q}$ and digital beamformer $\mathbf{V}$, the sum-rate maximization with respect to the SIM coefficients remains non-convex due to the cascaded structure of $\boldsymbol{\Psi}$. To exploit the layered architecture of the SIM, we adopt a BCD approach and optimize one layer at a time. Specifically, for the $\ell$-th layer, define
$ \mathbf{A}_\ell \triangleq \mathbf{\Phi}_L\mathbf{W}_L\cdots \mathbf{\Phi}_{\ell+1}\mathbf{W}_{\ell+1}$, and 
$ 
\mathbf{B}_\ell \triangleq \mathbf{W}_\ell \mathbf{\Phi}_{\ell-1}\mathbf{W}_{\ell-1}\cdots \mathbf{\Phi}_1\mathbf{W}_1\mathbf{V},
$
such that $ 
\boldsymbol{\Psi}\mathbf{V}=\mathbf{A}_\ell \mathbf{\Phi}_\ell \mathbf{B}_\ell .
$ 
Accordingly, for user $k$ and stream $i$, the effective scalar channel through layer $\ell$ is
$ 
g_{k,i}^{(\ell)} \triangleq \mathbf{h}_k^H(\mathbf{q})\mathbf{A}_\ell \mathbf{\Phi}_\ell \mathbf{B}_\ell \mathbf{e}_i,
$ 
where $\mathbf{e}_i$ denotes the $i$-th canonical basis vector. To further decouple the optimization, we update one diagonal entry of $\mathbf{\Phi}_\ell$ at a time. Let $\mathbf{\Phi}_\ell=\widetilde{\mathbf{\Phi}}_{\ell,n}+\phi_{\ell,n}\mathbf{E}_n$, where $\widetilde{\mathbf{\Phi}}_{\ell,n}$ contains all diagonal elements except the $n$-th one and $\mathbf{E}_n\triangleq\mathrm{diag}(\mathbf{e}_n)$. Then,
$ 
g_{k,i}^{(\ell)} = a_{k,i}^{(\ell,n)} + \phi_{\ell,n} b_{k,i}^{(\ell,n)},
$ 
where
$ 
a_{k,i}^{(\ell,n)} = \mathbf{h}_k^H(\mathbf{q})\mathbf{A}_\ell \widetilde{\mathbf{\Phi}}_{\ell,n}\mathbf{B}_\ell \mathbf{e}_i,\;
b_{k,i}^{(\ell,n)} = \mathbf{h}_k^H(\mathbf{q})\mathbf{A}_\ell \mathbf{E}_n \mathbf{B}_\ell \mathbf{e}_i$.
Accordingly, for fixed $\{\phi_{\ell,m}\}_{m\neq n}$, the SINR becomes
\begin{equation}
\gamma_k(\phi_{\ell,n})
=
\frac{|a_{k,k}^{(\ell,n)}+\phi_{\ell,n}b_{k,k}^{(\ell,n)}|^2}
{\sum_{i\neq k}|a_{k,i}^{(\ell,n)}+\phi_{\ell,n}b_{k,i}^{(\ell,n)}|^2+\sigma_k^2}.
\end{equation}
The resulting scalar optimization problem for $\phi_{\ell,n}$ is
\begin{equation}\label{eq:phi_scalar_prob}
\max_{|\phi_{\ell,n}|=1}
R_{\ell,n}(\phi_{\ell,n})
=
\sum_{k=1}^{K}\log_2(1+\gamma_k(\phi_{\ell,n})).
\end{equation}
 Let $S_k=|a_{k,k}^{(\ell,n)}+\phi_{\ell,n}b_{k,k}^{(\ell,n)}|^2$ and 
$I_k=\sum_{i\neq k}|a_{k,i}^{(\ell,n)}+\phi_{\ell,n}b_{k,i}^{(\ell,n)}|^2+\sigma_k^2$. Then
\begin{equation}
R_{\ell,n}=\sum_{k=1}^{K}\log_2\!\left(1+\frac{S_k}{I_k}\right)
=\sum_{k=1}^{K}\Big[\log_2(S_k+I_k)-\log_2(I_k)\Big].
\end{equation}
Using Wirtinger calculus, the derivative of a term $|\xi|^2$ with $\xi=a+\phi_{\ell,n}b$ satisfies
$\frac{\partial |\xi|^2}{\partial \phi_{\ell,n}^*}=\xi\, b^*$, and applying the chain rule yields
\begin{equation}
\frac{\partial (S_k+I_k)}{\partial \phi_{\ell,n}^*}
=
\sum_{i=1}^K \xi_{k,i} b_{k,i}^{(\ell,n)*},\quad
\frac{\partial I_k}{\partial \phi_{\ell,n}^*}
=
\sum_{i\neq k} \xi_{k,i} b_{k,i}^{(\ell,n)*}.
\end{equation}
Therefore, the gradient of $R_{\ell,n}$ with respect to $\phi_{\ell,n}^*$ is
\begin{equation}
\frac{\partial R_{\ell,n}}{\partial\phi_{\ell,n}^*}
=
\frac{1}{\ln2}
\sum_{k=1}^K
\left[
\frac{\sum_{i=1}^K \xi_{k,i} b_{k,i}^{(\ell,n)*}}{S_k+I_k}
-
\frac{\sum_{i\neq k}\xi_{k,i} b_{k,i}^{(\ell,n)*}}{I_k}
\right],
\end{equation}
where $\xi_{k,i}=a_{k,i}^{(\ell,n)}+\phi_{\ell,n}b_{k,i}^{(\ell,n)}$. Let $\phi_{\ell,n}^{(t)}$ denote the current iterate and define
\begin{equation}
g_{\ell,n}^{(t)} \triangleq 
\left.\frac{\partial R_{\ell,n}}{\partial \phi_{\ell,n}^*}\right|_{\phi_{\ell,n}=\phi_{\ell,n}^{(t)}} .
\end{equation}
Since $R_{\ell,n}$ is continuously differentiable, a quadratic minorizing surrogate can be constructed at $\phi_{\ell,n}^{(t)}$ under a suitable local Lipschitz bound on the gradient \cite{beck2017first}:
\begin{equation} \label{surrogate}
\widetilde R(\phi_{\ell,n}\mid \phi_{\ell,n}^{(t)})
=
R_{\ell,n}(\phi_{\ell,n}^{(t)})
+
2\Re\{g_{\ell,n}^{(t)*}(\phi_{\ell,n}-\phi_{\ell,n}^{(t)})\}
-\rho_{\ell,n}^{(t)}|\phi_{\ell,n}-\phi_{\ell,n}^{(t)}|^2,
\end{equation}
where $\rho_{\ell,n}^{(t)}>0$ is selected via backtracking\footnote{Starting from $\rho_{\ell,n}^{(t)}=\rho_0>0$, the candidate update \eqref{closed_form_SIM} is computed and accepted if $R_{\ell,n}(\phi_{\ell,n}^{(t+1)})\geq \widetilde R(\phi_{\ell,n}^{(t+1)}\mid\phi_{\ell,n}^{(t)})$; otherwise $\rho_{\ell,n}^{(t)}\leftarrow\tau\rho_{\ell,n}^{(t)}$, with $\tau>1$, and the test is repeated. Since $R_{\ell,n}$ is continuously differentiable with locally Lipschitz gradient on the compact set $|\phi_{\ell,n}|=1$, the test is satisfied after a finite number of trials, namely once $\rho_{\ell,n}^{(t)}$ exceeds the local Lipschitz constant of the gradient \cite{beck2017first}.}
so that $R_{\ell,n}(\phi_{\ell,n})\ge \widetilde R(\phi_{\ell,n}\mid \phi_{\ell,n}^{(t)})$, with equality at $\phi_{\ell,n}^{(t)}$. Thus, $\widetilde R$ is a valid minorizing surrogate, and maximizing it yields standard minorization-maximization (MM)-type updates \cite{razaviyayn2013unified,scutari2016parallel}.

\begin{theorem}
Maximizing $\widetilde R(\phi_{\ell,n}\mid \phi_{\ell,n}^{(t)})$ yields the update
\begin{equation} \label{closed_form_SIM}
\phi_{\ell,n}^{(t+1)}
=
e^{j\angle c_{\ell,n}^{(t)}},
\qquad
c_{\ell,n}^{(t)}=g_{\ell,n}^{(t)}+\rho_{\ell,n}^{(t)}\phi_{\ell,n}^{(t)} .
\end{equation}
Moreover,
$ 
R_{\ell,n}(\phi_{\ell,n}^{(t+1)})\ge R_{\ell,n}(\phi_{\ell,n}^{(t)})$. 
\end{theorem}
 \begin{proof}
Starting from the surrogate in \eqref{surrogate}, we maximize $\widetilde R(\phi_{\ell,n}\mid \phi_{\ell,n}^{(t)})$ over $|\phi_{\ell,n}|=1$. Discarding terms independent of $\phi_{\ell,n}$ and expanding $|\phi_{\ell,n}-\phi_{\ell,n}^{(t)}|^2$, while using $|\phi_{\ell,n}|=|\phi_{\ell,n}^{(t)}|=1$, the surrogate reduces to
\begin{equation}
\widetilde R(\phi_{\ell,n}\mid \phi_{\ell,n}^{(t)})
\equiv
2\Re\{\phi_{\ell,n}^* c_{\ell,n}^{(t)}\}+\text{const},
\end{equation}
where $c_{\ell,n}^{(t)}=g_{\ell,n}^{(t)}+\rho_{\ell,n}^{(t)}\phi_{\ell,n}^{(t)}$. Hence, we have $
\max_{|\phi_{\ell,n}|=1}\Re\{\phi_{\ell,n}^* c_{\ell,n}^{(t)}\}$, which can be optimally 
solved by phase alignment, yielding $\phi_{\ell,n}^{(t+1)}=e^{j\angle c_{\ell,n}^{(t)}}$. Since $\widetilde R$ is a global lower bound of $R_{\ell,n}$ and is tight at $\phi_{\ell,n}^{(t)}$, monotonic improvement follows \cite{razaviyayn2013unified}.
\end{proof}
Algorithm~1 summarizes the SIM design: given the current $\mathbf{V}$ and $\mathbf{q}$, it sweeps across the $L$ layers, updating each phase shift in closed form via \eqref{closed_form_SIM} with backtracked $\rho_{\ell,n}$, until the sum-rate improvement falls below $\varepsilon_{\rm sim}$; every update is non-decreasing by Theorem~1.
\begin{algorithm}[t]
\caption{SIM Optimization}
\begin{algorithmic}[1]
\Require $\mathbf{q}$, $\mathbf{V}$, $\{\mathbf{W}_\ell\}_{\ell=1}^{L}$, initial $\{\mathbf{\Phi}_\ell\}_{\ell=1}^{L}$, tolerance $\varepsilon_{\rm sim}$
\State Construct the current $\boldsymbol{\Psi}$ from $\{\mathbf{\Phi}_\ell\}_{\ell=1}^{L}$ and compute the initial sum rate $R_{\rm sim}^{(0)}$
\State Set the inner iteration index $r=0$
\Repeat
\For{$\ell=1,\ldots,L$}
\State Compute $\mathbf{A}_\ell=\mathbf{\Phi}_L\mathbf{W}_L\cdots\mathbf{\Phi}_{\ell+1}\mathbf{W}_{\ell+1}$
\State Compute $\mathbf{B}_\ell=\mathbf{W}_\ell\mathbf{\Phi}_{\ell-1}\mathbf{W}_{\ell-1}\cdots\mathbf{\Phi}_1\mathbf{W}_1\mathbf{V}$
\For{$n=1,\ldots,M$}
\State Compute $g_{\ell,n}^{(r)}$ and $\rho_{\ell,n}^{(r)}>0$ by backtracking
\State Form $c_{\ell,n}^{(r)}=g_{\ell,n}^{(r)}+\rho_{\ell,n}^{(r)}\phi_{\ell,n}^{(r)}$
\State Update $\phi_{\ell,n}^{(r+1)}=e^{j\angle c_{\ell,n}^{(r)}}$
\EndFor
\State Set $\mathbf{\Phi}_\ell=\mathrm{diag}([\phi_{\ell,1}^{(r+1)},\ldots,\phi_{\ell,M}^{(r+1)}]^T)$
\EndFor
\State Reconstruct $\boldsymbol{\Psi}$ from $\{\mathbf{\Phi}_\ell\}_{\ell=1}^{L}$ and compute $R_{\rm sim}^{(r+1)}$
\State $r\leftarrow r+1$
\Until{$|R_{\rm sim}^{(r)}-R_{\rm sim}^{(r-1)}|\le \varepsilon_{\rm sim}$}
\State \Return $\{\mathbf{\Phi}_\ell\}_{\ell=1}^{L}$
\end{algorithmic}
\end{algorithm}

\subsection{UAV 3D Position Optimization}
For fixed SIM coefficients and digital beamformer, define the effective transmit beamforming vectors $\mathbf{f}_i \triangleq \boldsymbol{\Psi}\mathbf{v}_{i}$, $i=1,\ldots,K$. Accordingly, the received signal component from stream $i$ at user $k$ is
\begin{equation}
z_{k,i}(\mathbf{q})=\mathbf{h}_k^H(\mathbf{q})\mathbf{f}_i
=\sum_{m=1}^{M} h_{k,m}^*(\mathbf{q})[\mathbf{f}_i]_m .
\end{equation}

Hence, the SINR at user $k$ becomes $\gamma_k(\mathbf{q})=\frac{|z_{k,k}(\mathbf{q})|^2}{\sum_{i\neq k}|z_{k,i}(\mathbf{q})|^2+\sigma_k^2}$.

The UAV-position optimization subproblem is
\begin{equation}\label{eq:q_rate_problem}
\max_{\mathbf{q}}\quad 
R(\mathbf{q})=\sum_{k=1}^K \log_2\!\big(1+\gamma_k(\mathbf{q})\big)\quad
\;\text{s.t.}\;\eqref{opt_problem_x}.
\end{equation}
Define $S_k(\mathbf{q})=|z_{k,k}(\mathbf{q})|^2$ and $I_k(\mathbf{q})=\sum_{i\neq k}|z_{k,i}(\mathbf{q})|^2+\sigma_k^2$, so that $\gamma_k=S_k/I_k$. Using $\nabla \log(1+x)=\frac{\nabla x}{1+x}$, the gradient of the sum rate is
\begin{equation}\label{eq:grad_rate_q}
\nabla_{\mathbf{q}}R(\mathbf{q})
=
\frac{1}{\ln 2}
\sum_{k=1}^{K}
\frac{
I_k\nabla_{\mathbf{q}}S_k
-
S_k\nabla_{\mathbf{q}}I_k
}{
I_k(S_k+I_k)
}.
\end{equation}
Using the chain rule,
\begin{align}
\nabla_{\mathbf{q}}S_k
&=
2\Re\!\left\{
z_{k,k}^*\nabla_{\mathbf{q}}z_{k,k}
\right\},
\nabla_{\mathbf{q}}I_k
=
\sum_{i\neq k}
2\Re\!\left\{
z_{k,i}^*\nabla_{\mathbf{q}}z_{k,i}
\right\},
\end{align}
where
\begin{equation}
\nabla_{\mathbf{q}}z_{k,i}
=
\sum_{m=1}^{M}[\mathbf{f}_i]_m
\nabla_{\mathbf{q}}h_{k,m}^*(\mathbf{q}).
\end{equation}
From \eqref{canale}, since $d_{k,m}(\mathbf{q})=\|\mathbf{q}+\mathbf{p}_m-\mathbf{u}_k\|$, we have
\begin{equation}
\nabla_{\mathbf{q}}h_{k,m}^*(\mathbf{q})
=
h_{k,m}^*(\mathbf{q})
\left(
-\frac{\alpha}{2d_{k,m}(\mathbf{q})}
+
j\frac{2\pi}{\lambda}
\right)
\frac{\mathbf{q}+\mathbf{p}_m-\mathbf{u}_k}{d_{k,m}(\mathbf{q})}.
\end{equation}

Since \eqref{eq:q_rate_problem} is non-concave, we update $\mathbf{q}$ via projected gradient ascent (PGA):
\begin{equation}\label{eq:q_update}
\mathbf{q}^{(t+1)}
=
\Pi
\left(
\mathbf{q}^{(t)}
+
\mu_t
\nabla_{\mathbf{q}}R(\mathbf{q}^{(t)})
\right),
\end{equation}
where $\mu_t$ is a suitable step size (e.g., obtained via backtracking) and $\Pi(\cdot)$ denotes projection onto the feasible region.
Algorithm~2 summarizes the overall design: at each outer iteration, the ZF beamformer is recomputed via \eqref{digi_solu}, the SIM layers are optimized by Algorithm~1, and the UAV position is refined via the PGA step \eqref{eq:q_update}, until the sum-rate improvement falls below $\epsilon$ or $T_{\max}$ iterations are reached.

\begin{algorithm}[t]
\caption{Joint Optimization Framework}
\begin{algorithmic}[1]
\Require Initialize feasible $\mathbf{q}^{(0)}$, $\{\mathbf{\Phi}_\ell^{(0)}\}_{\ell=1}^{L}$, tolerance $\epsilon$
\State Set $t=0$ and compute the initial sum rate $R^{(0)}$
\Repeat
\State Update $\mathbf{V}^{(t+1)}$ with \eqref{digi_solu}
\State Update $\{\mathbf{\Phi}_\ell^{(t+1)}\}_{\ell=1}^{L}$ using Algorithm~1 
\State Construct $\boldsymbol{\Psi}^{(t+1)}=\mathbf{\Phi}_L^{(t+1)}\mathbf{W}_L\cdots \mathbf{\Phi}_1^{(t+1)}\mathbf{W}_1$
\State Compute $\nabla_{\mathbf{q}}R(\mathbf{q}^{(t+1)})$ from \eqref{eq:grad_rate_q}
\State Update the UAV position with \eqref{eq:q_update} 
\State $t\leftarrow t+1$
\Until{$|R^{(t)}-R^{(t-1)}|\le \epsilon$ or $t=T_{\max}$}
\State \Return $\mathbf{q}^{(t)}$, $\mathbf{V}^{(t)}$, and $\{\mathbf{\Phi}_\ell^{(t)}\}_{\ell=1}^{L}$
\end{algorithmic}
\end{algorithm}
  \subsection{Convergence}

The proposed framework alternates over the blocks $\mathbf{V}$, $\{\mathbf{\Phi}_\ell\}_{\ell=1}^{L}$, and $\mathbf{q}$, where each block update yields a non-decreasing objective: the MM-based SIM updates by Theorem~1, and the PGA step through the backtracking line search. 
Therefore, the objective sequence $\{R^{(t)}\}$ is non-decreasing and, since the feasible set is compact, upper bounded; hence, by the monotone convergence theorem, it converges to a finite value \cite{beck2017first}. Regarding the iterates, since the UAV subproblem \eqref{eq:q_rate_problem} is non-concave, convergence to a jointly stationary point of \eqref{opt_problem} cannot be claimed in general. Nevertheless, as the rate function is continuously differentiable with Lipschitz gradient over the compact box \eqref{opt_problem_x}, PGA with backtracking ensures that every limit point of the UAV updates satisfies the first-order optimality condition of \eqref{eq:q_rate_problem} \cite{beck2017first}; combined with the tight MM surrogates \cite{razaviyayn2013unified}, the limit points of the generated sequence satisfy blockwise first-order optimality conditions. The convergence behaviour is illustrated in Fig.~\ref{fig:conv}.
 \begin{figure}[t]
 	\centering
 	\subfloat[Convergence behaviour with $P_{\max}=15$ dBm.%
 	\label{fig:conv}]{%
 		\includegraphics[width=0.48\linewidth]{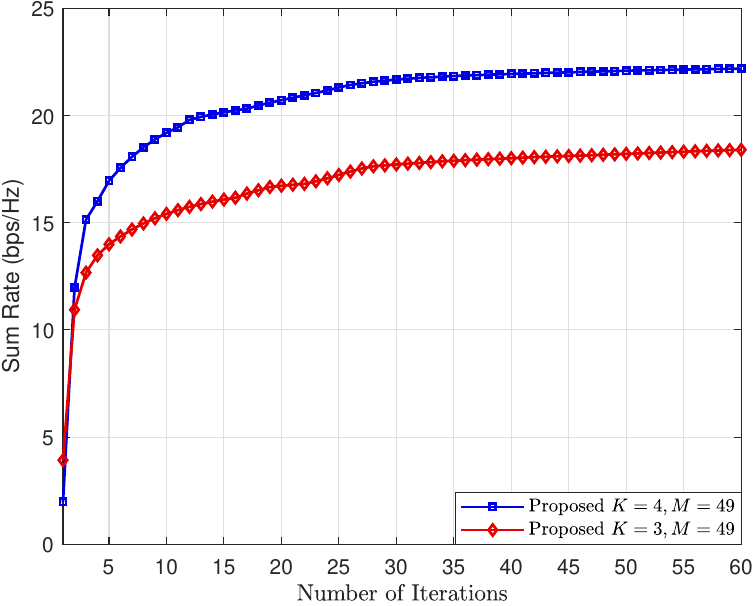}
 	}
 	\hfill
 	\subfloat[Sum rate versus $P_{\max}$ with $L=3$.%
 	\label{fig:pow}]{%
 		\includegraphics[width=0.48\linewidth]{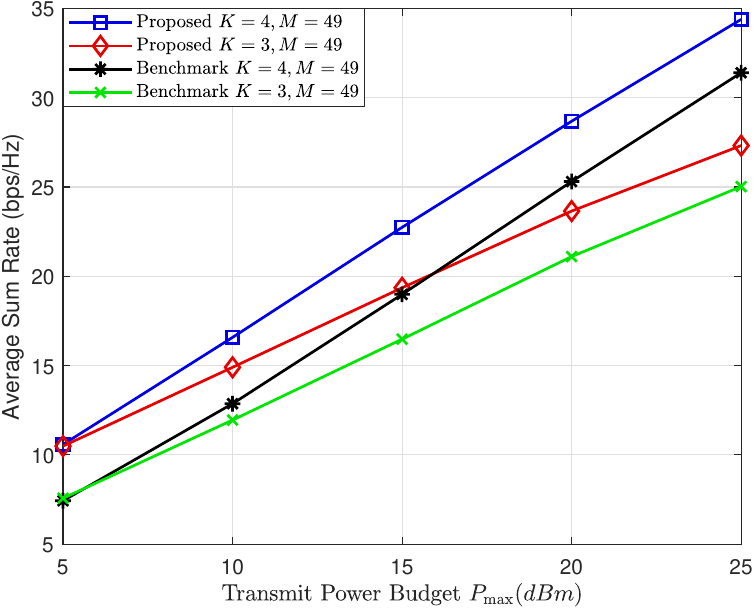}
 	}
 	\caption{Performance of the SIM-assisted UAV system: 
 		(a) convergence behaviour and (b) sum rate versus $P_{\max}$.}
 	\label{fig:combined2}
 \end{figure}
We consider a downlink system serving $K=3$ or $4$ single-antenna users randomly distributed over a circle of radius $30$\,m, with $N=K$ RF chains. The UAV moves within the horizontal region $[-30\,\text{m},30\,\text{m}]\times[-30\,\text{m},30\,\text{m}]$ with altitude $H\in[30\,\text{m},70\,\text{m}]$. We set $\rho_0=10^{-3}$ and $\tau=2$ in all simulations. The carrier frequency is set to $f_c=3\,\text{GHz}$ with a distance-dependent path-loss model with exponent $\alpha=2.2$. The noise power at the users is set to $\sigma_k^2=-90\,\text{dBm}$. The distance between the layers is set to $\lambda$ and the phase configuration of the layers is randomly initialized. The UAV is initialized at the centroid of the user locations, while the digital beamformers are initialized using the dominant eigenvectors of the effective initial channel covariance matrix. Results are averaged over multiple independent user realizations. The proposed scheme (\emph{proposed}) is compared with a conventional hybrid beamforming scheme employing a single-layer metasurface at the antenna front-end (\emph{benchmark}).

Fig.~\ref{fig:pow} depicts the sum rate versus the transmit power $P_{\max}$: both schemes benefit from increased power, but the proposed method consistently outperforms the benchmark, with a gap that grows with $P_{\max}$. This is because the additional wave-domain degrees of freedom of the multi-layer SIM yield a better-conditioned effective channel $\bar{\mathbf{H}}$, reducing the power penalty of the ZF inversion and increasing the effective channel gains toward the users. Since this beamforming advantage scales the received SNR multiplicatively, the corresponding sum-rate offset becomes increasingly visible as $P_{\max}$ grows, whereas the single-layer benchmark exploits the extra power less efficiently. Fig.~\ref{fig:L} plots the sum rate versus the number of layers $L$ with $M=49$: increasing $L$ enhances the proposed scheme through richer wave-domain processing, with gains saturating for larger $L$, while the single-layer benchmark remains unaffected. Finally, Fig.~\ref{fig:M} shows that, for $L=3$, the performance increases significantly with $M$ owing to the enhanced array gain and spatial resolution.
These results reveal a design trade-off between stacking depth and aperture size: for fixed $M$ the gain from adding layers saturates (Fig.~\ref{fig:L}), whereas for fixed $L$ the sum rate keeps increasing with $M$ (Fig.~\ref{fig:M}). Hence, under a fixed budget of $LM$ meta-atoms, a moderate $L$ combined with a sufficiently large aperture per layer better exploits the potential of SIM-assisted UAV systems.
 \begin{figure}[t]
 	\centering
 	\subfloat[Sum rate versus $L$ with $P_{\max}=15$ dBm.%
 	\label{fig:L}]{%
 		\includegraphics[width=0.48\linewidth]{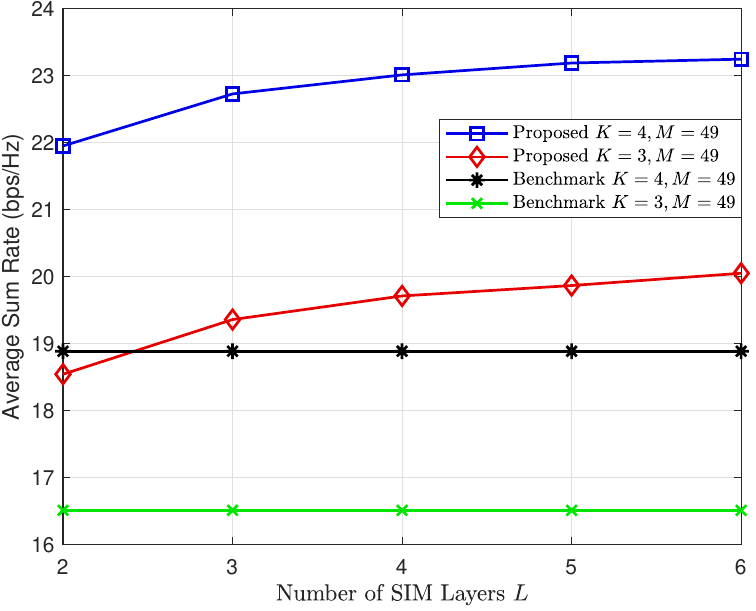}
 	}
 	\hfill
 	\subfloat[Sum rate versus $M$ with $P_{\max}=15$ dBm.%
 	\label{fig:M}]{%
 		\includegraphics[width=0.48\linewidth]{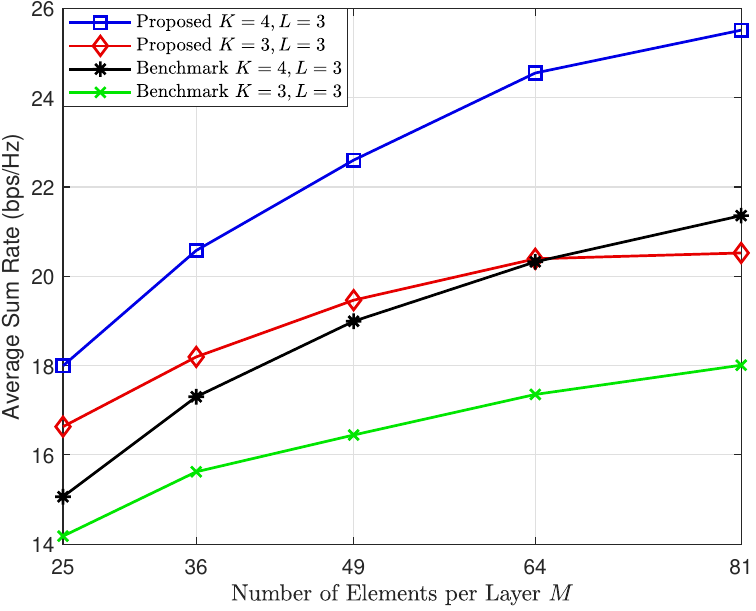}
 	}
 	\caption{Sum-rate performance versus the number of layers $L$ and the number of meta-atoms per layer $M$.}
 	\label{fig:combined}
 \end{figure}
 \section{Conclusion}

In this letter, we investigated SIM-mounted UAV multi-user communications and proposed a joint optimization framework for hybrid beamforming and UAV positioning. By exploiting the problem structure, we developed an efficient alternating algorithm with closed-form updates and gradient-based refinement, with a monotonically non-decreasing objective. Numerical results demonstrated significant performance gains over conventional single-layer RIS architectures, highlighting the benefits of multi-layer wave-domain processing.

\ifCLASSOPTIONcaptionsoff
  \newpage
\fi

{\footnotesize
\bibliographystyle{IEEEtran}
\def\baselinestretch{1}
\bibliography{main}}
\end{document}